\documentclass[submission,copyright,creativecommons]{eptcs}
\providecommand{\event}{SOS 2011} 

\usepackage[all]{xy}
\usepackage[pdftex]{graphics}
\usepackage{graphicx} 
\usepackage{latexsym}
\usepackage{pxfonts}
\usepackage{pst-all}
\usepackage{pstricks}
\usepackage{amsfonts}
\usepackage{amssymb} 
\usepackage{float} 
\usepackage{wrapfig} 
\usepackage{slashbox}
\usepackage{multirow}
\usepackage{txfonts}

\title{On the Unification of Process Semantics: Logical Semantics}
\author{David Romero-Hern\'andez \quad\qquad David de Frutos-Escrig
\institute{Facultad CC. Matem{\'a}ticas, Universidad Complutense de Madrid\\ Madrid, Spain} \institute{Departamento de Sistemas Inform\'aticos y Computaci\'on
\thanks{This work was partially supported by the Spanish projects TESIS (TIN2009-14312-C02-01), DESAFIOS10 (TIN2009-14599-C03-01) and PROMETIDOS S2009 / TIC-1465. }\\} \email{dromeroh@pdi.ucm.es \quad\qquad defrutos@sip.ucm.es} }
\def\titlerunning{On the Unification of Process Semantics: Logical Semantics}
\def\authorrunning{D. Romero-Hern\'andez \& D. de Frutos-Escrig}

\newcommand{\formyacc}[1]{
$\varphi \in \mathcal{L}^{\prime}_{#1},  \hspace{0.1cm}a \in \emph{Act} \Rightarrow a\varphi \in \mathcal{L}^{\prime}_{#1}$}

\newcommand{\conjU}[2]{
$\sigma, \sigma_j \in \mathcal{L}^{\prime}_#2  \forall j\in J\Rightarrow (\sigma\wedge\bigwedge_{j \in J}\neg \sigma_j\top) \in \mathcal{L}_{#1}$}

\newcommand{\conjsim}[2]{
$\sigma \in \mathcal{L}^{\prime}_#2 \Rightarrow \sigma\in \mathcal{L}^{\prime}_{#1}$; $\sigma \in \mathcal{L}^{\prime}_#2 \Rightarrow \neg \sigma\in \mathcal{L}^{\prime}_{#1}$}
\newcommand{\conj}[2]{
$\sigma \in \mathcal{L}^{\equiv}_#2 \Rightarrow \sigma \in \mathcal{L}^{\prime}_{#1}$}
\newcommand{\conjparticular}[2]{
$X_1, X_2 \subseteq \mathcal{L}_#2\Rightarrow (\bigwedge_{a\top \in X_1} a\top \wedge \bigwedge_{b\top \in X_2}\neg b\top) \in \mathcal{L}^{\prime}_{#1}$}
\newcommand{\conjneg}[2]{
$\sigma \in \mathcal{L}^{\neg}_#2 \Rightarrow \sigma \in \mathcal{L}^{\prime}_{#1}$}
\newcommand{\conjmo}[2]{
$\sigma \in \mathcal{L}^{\surd}_#2 \Rightarrow \sigma \in \mathcal{L}^{\prime}_{#1}$}
\newcommand{\conjnegparticular}[2]{
$X_1 \subseteq \mathcal{L}_#2 \Rightarrow (\bigwedge_{a\top \in X_1}\neg a\top) \in \mathcal{L}^{\prime}_{#1}$}
\newcommand{\conjyform}[2]{
$\varphi \in \mathcal{L}^{\prime}_{#1}, \hspace{0.1cm}  \sigma \in \mathcal{L}^{\equiv}_#2 \Rightarrow \sigma \wedge \varphi \in \mathcal{L}^{\prime}_{#1}$}
\newcommand{\conjyformparticular}[2]{
$\varphi \in \mathcal{L}^{\prime}_{#1}, X_1, X_2 \subseteq \mathcal{L}_#2\Rightarrow (\bigwedge_{a\top \in X_1} a\top \wedge \bigwedge_{b\top \in X_2}\neg b\top) \wedge \varphi\in \mathcal{L}^{\prime}_{#1}$}
\newcommand{\conjyformneg}[2]{
$\varphi \in \mathcal{L}^{\prime}_{#1},  \hspace{0.1cm} \sigma \in \mathcal{L}^{\neg}_#2 \Rightarrow \sigma \wedge \varphi \in \mathcal{L}^{\prime}_{#1}$}
\newcommand{\conjyformmo}[2]{
$\varphi \in \mathcal{L}^{\prime}_{#1},  \hspace{0.1cm} \sigma \in \mathcal{L}^{\surd}_#2 \Rightarrow \sigma \wedge \varphi \in \mathcal{L}^{\prime}_{#1}$}
\newcommand{\conjyformnegparticular}[2]{
$\varphi \in \mathcal{L}^{\prime}_{#1}, X_1 \subseteq \mathcal{L}_#2 \Rightarrow (\bigwedge_{a\top \in X_1}\neg a\top)\wedge \varphi \in \mathcal{L}^{\prime}_{#1}$}

\newcommand{\bran}[2]{#1; $\varphi_i \in \mathcal{L}^{\prime}_{#2}\forall i \in I\Rightarrow \bigwedge_{i\in I} \varphi_i \in \mathcal{L}^{\prime}_{#2}$; \formyacc{#2};}
\newcommand{\brannew}[2]{#1; $\varphi_i \in \mathcal{L}^{\prime}_{#2}\forall i \in I\Rightarrow \bigwedge_{i\in I} \varphi_i \in \mathcal{L}^{\prime}_{#2}$; \formyacc{#2}.}

\newcommand{\dbran}[2]{$\top \in \mathcal{L}^{\prime}_{#2}$; #1; $X\subseteq \emph{Act},  \varphi_a \in \mathcal{L}^{\prime}_{#2}\forall a \in X\Rightarrow \bigwedge_{a\in X} a\varphi_a \in \mathcal{L}^{\prime}_{#2}$.}

\newcommand{\lin}[2]{$\top \in \mathcal{L}^{\prime}_{#2}$; #1; \formyacc{#2}\ .}

\newcommand{\linf}[2]{$\top \in \mathcal{L}^{\prime}_{#2}$; \conjyformneg{#2}{#1}; \formyacc{#2}\ .}
\newcommand{\linfmo}[2]{$\top \in \mathcal{L}^{\prime}_{#2}$; \conjyformmo{#2}{#1}; \formyacc{#2}\ .}
\newcommand{\linfparticular}[2]{$\top \in \mathcal{L}^{\prime}_{#2}$; \conjyformnegparticular{#2}{#1}; \formyacc{#2}\ .}

\newcommand{\linc}[2]{$\top \in \mathcal{L}^{\prime}_{#2}$; \conj{#2}{#1}; \formyacc{#2}\ .}
\newcommand{\lincparticular}[2]{$\top \in \mathcal{L}^{\prime}_{#2}$; \conjparticular{#2}{#1}; \formyacc{#2}\ .}

\newcommand{\linfc}[2]{$\top \in \mathcal{L}^{\prime}_{#2}$; \conjneg{#2}{#1}; \formyacc{#2}\ .}
\newcommand{\linfcmo}[2]{$\top \in \mathcal{L}^{\prime}_{#2}$; \conjmo{#2}{#1}; \formyacc{#2}\ .}
\newcommand{\linfcparticular}[2]{$\top \in \mathcal{L}^{\prime}_{#2}$; \conjnegparticular{#2}{#1}; \formyacc{#2}\ .}

\newcommand{\conjsimold}[2]{
$\sigma \in \mathcal{L}_#2 \hspace {0.05cm} \Rightarrow \sigma\in \mathcal{L}^{\prime}_{#1}$; $\sigma \in \mathcal{L}_#2 \Rightarrow \neg \sigma\in \mathcal{L}^{\prime}_{#1}$}

\newenvironment{proof}{\par\addvspace{\bigskipamount}
\noindent\textit{\textbf{Proof.}}\ }{\par\addvspace{\bigskipamount}}

\begin{document}
\newtheorem{definition}{Definition}
\newtheorem{theorem}{Theorem}
\newtheorem{Proposition}{Proposition}
\newtheorem{lemma}{Lemma}
\newtheorem{example}{Example}
\newtheorem{remark}{Remark}

\maketitle

\begin{abstract}
We continue with the task of obtaining a unifying view of process semantics by considering in this case the logical characterization of the semantics. We start by considering the classic linear time-branching time spectrum developed by R.J. van Glabbeek. He provided a logical characterization of most of the semantics in his spectrum but, without following a unique pattern. In this paper, we present a uniform logical characterization of all the semantics in the enlarged spectrum. The common structure of the formulas that constitute all the corresponding logics gives us a much clearer picture of the spectrum, clarifying the relations between the different semantics, and allows us to develop generic proofs of some general properties of the semantics.
\end{abstract}

\section{Introduction} \label{introduction}
The definition of the semantics for concurrent / non-deterministic processes is a delicate question. As soon as the effect of non-determinism is taken into account we have to decide to which extent we will do so. Trace semantics, which were adequate for deterministic systems, obviously do not consider non-determinism at all. Instead, bisimulation semantics captures all the information induced by the choices at the observed process. There are different semantics for processes in the literature. The most popular of them were collected in van Glabbeek's linear time-branching time spectrum \cite{Gla01}, after being introduced along the years by different authors. At the abstract level a semantics is just an equivalence relation (or a preorder) between processes. These can be defined by choosing between different frameworks for the different semantics, so we have operational, observational, testing, logical and equational semantics. 

In \cite{Gla01} we find the famous picture of the ltbt-spectrum (Figure \ref{spectrum_Gla}) and descriptions of all the semantics in it including observational / testing, logical and equational (when possible) characterizations. Certainly, the basic elements used in the characterizations for a given framework are somewhat related, but a more systematic approach is desirable.
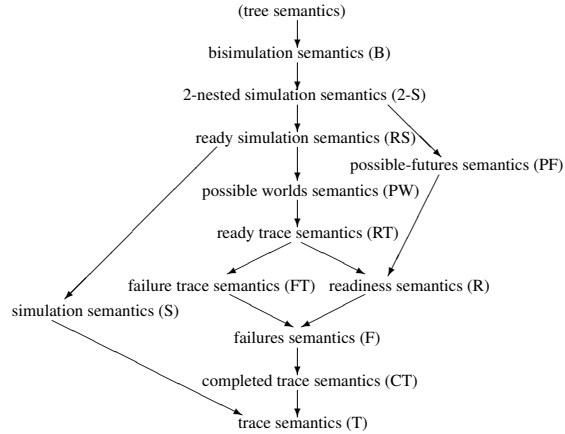
\begin{figure}[!]
\begin{center}
\scalebox{0.8}{\scriptsize
\begin{picture}(270,195)
\put(107,195){(tree semantics)}
     \put(135,193){\vector(0,-1){13}}
\put(93,175){bisimulation semantics (B)}
     \put(135,173){\vector(0,-1){13}}
\put(81,155){2-nested simulation semantics (2-S)}
     \put(135,153){\vector(0,-1){13}}
\put(87,135){ready simulation semantics (RS)}
     \put(135,133){\vector(0,-1){16}}
\put(90,110){possible worlds semantics (PW)}
     \put(135,108){\vector(0,-1){13}}
\put(97,90){ready trace semantics (RT)}
     \put(132,88){\vector(-2,-1){30}}
     \put(137,88){\vector(2,-1){30}}
\put(55,65){failure trace semantics (FT)}
\put(150,65){readiness semantics (R)}
     \put(103,63){\vector(2,-1){30}}
     \put(167,63){\vector(-2,-1){30}}
\put(105,40){failures semantics (F)}
     \put(135,38){\vector(0,-1){13}}
\put(90,20){completed trace semantics (CT)}
     \put(135,18){\vector(0,-1){13}}
\put(107,0){trace semantics (T)}

     \put(178,153){\vector(1,-1){25}}
\put(160,122){possible-futures semantics (PF)}
     \put(202,120){\vector(-1,-2){24}}

     \put(97,133){\vector(-1,-1){72}}
\put(0,53){simulation semantics (S)}
     \put(19,51){\vector(2,-1){90}}
\end{picture}}
\caption{The ltbt-spectrum}\label{spectrum_Gla}
\vspace{-0.7cm}
\end{center}
\end{figure}
\normalsize
In \cite{fgp09_equational,fgp09_observational}, a unified presentation of both the observational and the equational semantics has been developed, and it has been shown how the generic definitions allow to relate both without repeating similar arguments.

In this paper we present a unified view of the logical semantics by showing how different subsets of the Hennessy-Milner logic HML \cite{hm85} characterize each of the semantics in the spectrum. Certainly, the logical characterizations presented in \cite{Gla01} were also subsets of HML; however in that paper the author looked for sets of formulas as simple (and hence as small) as possible, probably driven by the idea that a smaller set of formulas would make any study based on it simpler. Instead, we will follow the opposite approach. Formally speaking, for each semantics defined by a preorder $\prec$ we have a (larger) language $\mathcal{L} \subseteq HML$ characterizing it, that is defined by $\varphi \in \mathcal{L}$ $\Leftrightarrow$ $((p \prec q \wedge p \models \varphi) \Rightarrow q \models \varphi)$. However, it is not easy (nor specially illustrative) to look for the whole set of formulas characterizing each of the semantics: we will consider sufficiently large families defined in a simple way, that provide more natural characterizations which immediately show the relations between the different semantics. For instance, whenever a semantics is finer than other, the logic characterizing the first will contain that for the latter.

As already happened in \cite{fgp09_equational,fgp09_observational}, our unified logical semantics will provide an \emph{enlarged spectrum} (Figure \ref{spectrum_fgp}) with a clearer structure and additional nodes which correspond to new semantics that in some cases have been also defined using different frameworks by several authors. In particular, we will show the logical characterization of revivals semantics introduced by B. Roscoe in \cite{Ros09}, that was already axiomatized in \cite{fgp09_equational}. 
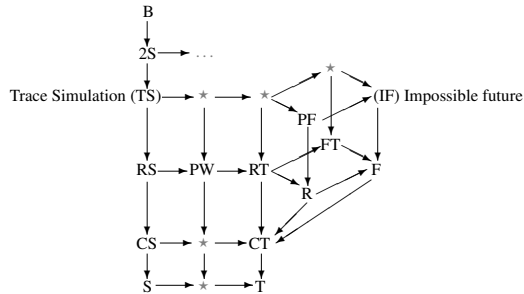
\begin{figure}[h]
\begin{center}
\scalebox{0.8}{\scriptsize
\begin{picture}(112,128)
\put(0,0){S} \put(26,1){\textcolor{gray}{$\star$}} \put(53,0){T}
\put(-3,20){CS} \put(26,21){\textcolor{gray}{$\star$}} \put(50,20){CT}
\put(-3,55){RS} \put(22,55){PW} \put(50,55){RT} \put(75,44){R} \put(84,67){FT} \put(108,55){F}
\put(-63,90){Trace Simulation (TS)} \put(26,91){\textcolor{gray}{$\star$}} \put(55,91){\textcolor{gray}{$\star$}} \put(73,79){PF} \put(86,104){\textcolor{gray}{$\star$}} \put(109,90){(IF) Impossible future}
\put(-2,110){2S} \put(25,111){\textcolor{gray}{\dots}}
\put(0,130){B}
\put(6,3){\vector(1,0){18}} \put(34,3){\vector(1,0){18}}
\put(2,18){\vector(0,-1){12}} \put(8,23){\vector(1,0){15}} \put(34,23){\vector(1,0){15}}
\put(29,18){\vector(0,-1){12}}
\put(56,18){\vector(0,-1){12}}
\put(2,52){\vector(0,-1){25}} \put(7,57){\vector(1,0){14}} \put(35,57){\vector(1,0){14}}
                     \put(61,57){\vector(2,1){23}}  \put(82,46){\vector(2,1){24}}
                     \put(61,57){\vector(2,-1){13}}  \put(94,69){\vector(2,-1){14}}
\put(29,52){\vector(0,-1){25}}
\put(56,52){\vector(0,-1){25}}
\put(77.5,42){\vector(-1,-1){15}} \put(108,53){\vector(-3,-2){45}}
\put(2,87){\vector(0,-1){25}} \put(9,92){\vector(1,0){15}} \put(34,92){\vector(1,0){17}}
                     \put(61,92){\vector(2,1){23}}  \put(85,81){\vector(2,1){23}}
                     \put(61,91){\vector(2,-1){12}}  \put(94,104){\vector(2,-1){14}}
\put(29,87){\vector(0,-1){25}}
\put(56,87){\vector(0,-1){25}} \put(78,78){\vector(0,-1){28}} \put(89,102){\vector(0,-1){28}} \put(111,87){\vector(0,-1){25}}
\put(2,108){\vector(0,-1){12}} \put(7,113){\vector(1,0){15}}
\put(2,128){\vector(0,-1){12}}
\end{picture}}
\caption{(A part of) the enlarged spectrum}\label{spectrum_fgp}
\vspace{-0.7cm}
\end{center}
\end{figure}
\normalsize

Moreover, we ``discover'' in this paper the semantics of minimal readies: it was not included in the previous version of the \emph{enlarged spectrum} because the development of the observational and equational frameworks did not detect its existence, while now in the logical framework its definition arises quite naturally.
Finally, we have also been able to discover a (minor) mistake in the classic logical characterization of one of the semantics in the original spectrum, Possible Worlds, that has been easily corrected when applying our uniform characterization. 

Due to lack of space we had to remove most of the proofs and also a part of the results. An extended version can be found at: \url{http://maude.sip.ucm.es/~miguelpt/papers/logsem.pdf}.

We strongly appreciate the comments and suggestions of the referees and those from Miguel Palomino, that have contributed to improve the presentation of the paper.
\section{Preliminaries}\label{preliminaries}
We will not repeat here the long list of original definitions of all the semantics in van Glabbeek's spectrum; please, take a look at \cite{Gla01}. The systematic classification of all these semantics using both observational and equational characterizations can be found at \cite{fgp09_equational,fgp09_observational}. All the semantics that we consider can be defined over arbitrary (possibly infinite) processes whose operational semantics is defined by means of a labelled transition system (lts) $\mathcal{P} = (Proc, Act, \rightarrow)$.  We will use the classical notation $p \stackrel{a}{\rightarrow} p^{\prime}$ to represent the transitions of processes. Moreover, it is also useful to have a syntactic notation for representing finite processes. We will use BCCSP \cite{Gla01, fgp09_equational}:

\begin{definition}
Given a set of actions Act, the set BCCSP(Act) of processes is defined by the BNF-grammar: $p::= \textbf{0} \  | \ ap \ | \ p+q $. We omit the known operational semantics of BCCSP, which can be found at \cite{Gla01, fgp09_equational}.
\end{definition}

The main ingredient in the classification of semantics, that of course was already present in the original spectrum, is the distinction between branching and linear time semantics. The most important branching semantics are the \textit{N}-constrained simulations that form the leftmost vertical line of the \emph{enlarged spectrum}. We like to call it the spine of the spectrum, because the rest of the semantics hang on (following the left to right lines) it. \textit{N}-constrained simulation were studied in a general and systematic way in \cite{fg08}.

\begin{definition}\label{sim:restringida}
Given a relation \textit{N} over BCCSP processes, an \textit{N}-constrained simulation is a relation \textit{$S_N$} such that $S_N \subseteq N$ and whenever \textit{p$S_N\,$q} if $p\stackrel{a}{\rightarrow}p^{\prime}$ then there exists $q^{\prime}$ with $q\stackrel{a}{\rightarrow}q^{\prime}$ and $p^{\prime}S_N\,q^{\prime}$.
We say that \textit{p} is \textit{N}-simulated by \textit{q}, or that \textit{q} \textit{N}-simulates \textit{p}, written \ \textit{p}$\sqsubseteq_{NS}$\textit{q}, when there exists an \textit{N}-constrained simulation $S_N$ such that $pS_N\,q$.
\end{definition}

Although in order to obtain \textit{N}-constrained similarities with good properties is not necessary for \textit{N} to be an equivalence relation, that happens in most of the interesting cases (including the most popular ones). For instance, Plain Simulations are just \textit{U}-constrained simulations, where \textit{U} is the universal relation $pUq$ $\forall p, q \in Proc$. Similarly, Ready Simulations can be defined by means of \textit{I}-simulations, with $pIq$ $\Leftrightarrow$ $I(p)=I(q)$ $\Leftrightarrow$ $(p\stackrel{a}{\rightarrow} \ \Leftrightarrow q \stackrel{a}{\rightarrow} \forall \ a \in Act)$; while Complete Simulations correspond to \textit{C}-simulations, taking $pCq$ $\Leftrightarrow$ $(\exists \ a \in Act \ p\stackrel{a}{\rightarrow} \ \Leftrightarrow \exists \ a\in Act \ q\stackrel{a}{\rightarrow})$.
Note that the Ready Simulation order is usually denoted by $\sqsubseteq_{RS}$, but when using our general notation $\sqsubseteq_{NS}$ we shall write instead $\sqsubseteq_{IS}$.

\subsection{Van Glabbeek's logical characterizations for process semantics}
Van Glabbeek also presented in \cite{Gla01} a logical characterization of the semantics in the (classical) linear time-branching time spectrum. These logics are sublanguages of the Hennessy-Milner logic \cite{hm85}, $\mathcal{L}_{HM}$, characterizing the bisimulation semantics in the general (possibly infinitary) case. 
\begin{definition}[\textbf{Hennessy-Milner logic, HML}]\label{HM_logic}
The set $\mathcal{L}_{HM}$ of Hennessy-Milner logical formulas is defined by:  if  $\varphi$, $\varphi_i \in \mathcal{L}_{HM}$ $\forall i \in I$ and $a \in Act$ then we have $\bigwedge_{i \in I} \varphi_i$, $a\varphi$, $\neg \varphi$ $\in \mathcal{L}_{HM}$.

For each labelled transition system $\mathbb{P}$, the satisfaction relation $\models\subseteq \mathbb{P} \times \mathcal{L}_{HM}$ is defined by:
\begin{itemize}
\item $p\models a\varphi$ if there exists $q\in \mathbb{P}:p\stackrel{a}{\rightarrow}q$ and $q \models \varphi$;
\item $p\models \bigwedge_{i\in I} \varphi_i$ if for all $i\in I: p \models \varphi_i$.
\item $p\models \neg \varphi$ if $p\nvDash \varphi$.
\end{itemize}
\end{definition}

Note that $\bigwedge_{i\in\emptyset}\varphi_i \in \mathcal{L}_{HM}$, and we have $p\models \bigwedge_{i\in\emptyset}\varphi_i$ for all \textit{p}. Therefore, in the following we will consider that $\top \in \mathcal{L}_{HM}$, where $\top$ is syntactic sugar for $\bigwedge_{i\in\emptyset}\varphi_i$. The finite version of this logic ($\mathcal{L}^{f}_{HM}$) uses binary conjunction $\wedge$ instead of the general conjunction $\bigwedge_{i \in I}$. It is well known that $\mathcal{L}^{f}_{HM}$ characterizes the bisimulation semantics between finite image processes, that are those that do not allow infinite branching for any action $a \in Act$ at any state.
Van Glabbeek uses $\mathcal{L}_{B}$ to refer to $\mathcal{L}_{HM}$ in \cite{Gla01}.

\begin{definition}
Any subset $\mathcal{L}$ of $\mathcal{L}_{HM}$ induces a logical semantics for processes, given by the preorder $\sqsubseteq_{\mathcal{L}}$: We have $p\sqsubseteq_{\mathcal{L}} q$ if, and only if, for all $\varphi \in \mathcal{L}$ $(p\models \varphi$ $\Rightarrow$ $q\models \varphi)$. We say that $\mathcal{L}$ and $\mathcal{L}^{\prime}$ are equivalent, and we write $\mathcal{L} \sim \mathcal{L}^{\prime}$, if they induce the same semantics, that is $\sqsubseteq_{\mathcal{L}}\,=\,\sqsubseteq_{\mathcal{L}^{\prime}}$.
\end{definition}

Table \ref{logic_table} contains the logical characterization of each of the semantics in van Glabbeek's spectrum: $\mathcal{L}_{Z}$ with $Z\in \{T, CT, F, FT, R, RT, PF, S, CS,$ $RS, 2S, PW, B\}$, denotes each of the logics; the dots indicate the clauses that we need to introduce to obtain the corresponding languages; and the boxes marked with $\mathcal{\nu}$ correspond to rules that could be added to $\mathcal{L}_{Z}$, but they would only introduce redundant formulas. The following connectives, which appear in the table, are not in $\mathcal{L}_{HM}$ but can be obtained as syntactic sugar: $$\widetilde{X}:=\bigwedge_{a\in X} \neg a \top \hspace{1.2cm} \widetilde{X}\varphi^{\prime}:=\widetilde{X} \wedge \varphi^{\prime} \hspace{1.2cm} 0:=\widetilde{Act}$$
$$\varphi_1 \wedge \varphi_2:=\bigwedge_{i\in \{1,2\}}\varphi_i \hspace{0.65cm} X:=\bigwedge_{a\in X}a\top \wedge \bigwedge_{a\not \in X} \neg a \top \hspace{0.7cm}X\varphi^{\prime}:=X \wedge \varphi^{\prime} \hspace{0.65cm} \widetilde{a}:=\neg a \top$$

\begin{table}[!]
\begin{center}
\scalebox{0.85}{\footnotesize
\begin{tabular}{|c|c|c|c|c|c|c|c|c|c|c|c|c||c|}
\hline
\backslashbox{Formulas}{Semantics ($\mathcal{Z}$)}& T & S & CT & CS & F & FT & R & RT & PW & RS & PF & 2S & B\\
\hline
$\top  \in \mathcal{L}_{\mathcal{Z}}$ & $\bullet$ & $\nu$ & $\bullet$ & $\nu$ & $\bullet$ & $\bullet$ & $\bullet$ & $\bullet$ & $\nu$ & $\nu$ & $\nu$ & $\nu$ & $\nu$\\
\hline
$\textbf{0}\in \mathcal{L}_{\mathcal{Z}}$ & & & $\bullet$ & $\bullet$ & $\nu$ & $\nu$ & $\nu$ & $\nu$ &$\nu$ & $\nu$ & $\nu$ & $\nu$ & $\nu$ \\
\hline
$\varphi \in \mathcal{L}_{\mathcal{Z}}, \hspace{0.075cm} a \in Act \Rightarrow$ &\multirow{2}{*}{$\bullet$} &\multirow{2}{*}{$\bullet$} &\multirow{2}{*}{$\bullet$} &\multirow{2}{*}{$\bullet$} &\multirow{2}{*}{$\bullet$} &\multirow{2}{*}{$\bullet$} &\multirow{2}{*}{$\bullet$} &\multirow{2}{*}{$\bullet$} &\multirow{2}{*}{$\nu$} &\multirow{2}{*}{$\bullet$} &\multirow{2}{*}{$\bullet$} & \multirow{2}{*}{$\bullet$} & \multirow{2}{*}{$\bullet$}\\
$a \varphi \in \mathcal{L}_{\mathcal{Z}}$& & & & & & & & & & & & &\\
\hline
$X \subseteq Act \Rightarrow$ & & & & & \multirow{2}{*}{$\bullet$} & \multirow{2}{*}{$\nu$} & \multirow{2}{*}{$\nu$} & \multirow{2}{*}{$\nu$} & \multirow{2}{*}{$\nu$} &\multirow{2}{*}{$\nu$} &\multirow{2}{*}{$\nu$} &\multirow{2}{*}{$\nu$} &\multirow{2}{*}{$\nu$} \\
$\widetilde{X} \in \mathcal{L}_{\mathcal{Z}}$ & & & & & & & & & & & & & \\
\hline
$X \subseteq Act \Rightarrow$ & & & & & & & \multirow{2}{*}{$\bullet$} &  \multirow{2}{*}{$\nu$} & \multirow{2}{*}{$\bullet$} &\multirow{2}{*}{$\bullet$}& \multirow{2}{*}{$\nu$} &\multirow{2}{*}{$\nu$} & \multirow{2}{*}{$\nu$}\\
$X \in \mathcal{L}_{\mathcal{Z}}$ & & & & & & & & & & & & &\\
\hline
$\varphi \in \mathcal{L}_{\mathcal{Z}}, \hspace{0.075cm} X \subseteq Act \Rightarrow$ & & & & & & \multirow{2}{*}{$\bullet$} & & \multirow{2}{*}{$\nu$} & \multirow{2}{*}{$\nu$}& \multirow{2}{*}{$\nu$} & & \multirow{2}{*}{$\nu$} & \multirow{2}{*}{$\nu$} \\
$\widetilde{X}\varphi \in \mathcal{L}_{\mathcal{Z}}$ & & & & & & & & & & & & &\\
\hline
$\varphi \in \mathcal{L}_{\mathcal{Z}}, \hspace{0.075cm} X \subseteq Act \Rightarrow$ & & & & & & & & \multirow{2}{*}{$\bullet$}  &\multirow{2}{*}{$\nu$} &\multirow{2}{*}{$\nu$} & &\multirow{2}{*}{$\nu$} & \multirow{2}{*}{$\nu$} \\
$X\varphi \in \mathcal{L}_{\mathcal{Z}}$ & & & & & & & & & & & & &\\
\hline
$\varphi_i \in \mathcal{L}_{\mathcal{Z}} \hspace{0.075cm} \forall i \in I \Rightarrow$ & & \multirow{2}{*}{$\bullet$} & &\multirow{2}{*}{$\bullet$} & & & & & &\multirow{2}{*}{$\bullet$} & & \multirow{2}{*}{$\bullet$} & \multirow{2}{*}{$\bullet$}\\
$\bigwedge_{i \in I} \varphi_i \in \mathcal{L}_{\mathcal{Z}}$ & & & & & & & & & & & & &\\
\hline
$X\subseteq Act, \hspace{0.075cm} \varphi_a \in \mathcal{L}_{PW} \hspace{0.075cm} \forall a \in X \Rightarrow$ & & & & & & & & & \multirow{2}{*}{$\bullet$}  &\multirow{2}{*}{$\nu$} & & \multirow{2}{*}{$\nu$} & \multirow{2}{*}{$\nu$}\\
$\bigwedge_{a \in X} a\varphi_a \in \mathcal{L}_{\mathcal{Z}}$ & & & & & & & & & & & & &\\
\hline
$\varphi_i, \varphi_j \in \mathcal{L}_{T} \hspace{0.075cm} \forall i \in I \hspace{0.05cm} \forall j \in J \Rightarrow$ & & & & & & & & & & &\multirow{2}{*}{$\bullet$} & \multirow{2}{*}{$\nu$}  & \multirow{2}{*}{$\nu$}\\
$\bigwedge_{i \in I} \varphi_i \wedge \bigwedge_{j \in J} \neg \varphi_j \in \mathcal{L}_{\mathcal{Z}}$ & & & & & & & & & & & & &\\
\hline
$\varphi \in \mathcal{L}_{S} \Rightarrow$ & & & & & & & & & & & & \multirow{2}{*}{$\bullet$} & \multirow{2}{*}{$\nu$}\\
$\neg \varphi \in \mathcal{L}_{\mathcal{Z}}$ & & & & & & & & & & & & &\\
\hline
$\varphi \in \mathcal{L}_{\mathcal{Z}} \Rightarrow$ & & & & & & & & & & & & &\multirow{2}{*}{$\bullet$}\\
$\neg \varphi \in \mathcal{L}_{\mathcal{Z}}$ & & & & & & & & & & & & &\\
\hline
\end{tabular}}
\vspace{1ex}
\caption{Van Glabbeek's logical characterizations for the semantics in the ltbt-spectrum} \label{logic_table}
\vspace{-0.7cm}
\end{center}
\end{table}

Disjunction does not appear in $\mathcal{L}_{HM}$, and therefore neither in any of the logics $\mathcal{L}_{Z}$ characterizing the semantics in the linear time-branching time spectrum. It is probably folklore that it can be added in all cases without changing the expressive power of each of these logics, but since we have not found a clear statement in this direction in any of our references, next we establish the result and comment on its proof. 

\begin{Proposition}\label{disyuncion_teo}
If we define $\mathcal{L}_Z^{\vee}$ with $Z\in \{T, CT, F, FT, R, RT, PF, S, CS,$ $RS, 2S, PW, B\}$, by adding the clause $\sigma_i \in \mathcal{L}_Z^{\vee} \hspace{0.2cm} \forall i \in I \Rightarrow \bigvee_{i\in I}\sigma_i \in \mathcal{L}_Z^{\vee}$ to the clauses which define each semantics $\mathcal{L}_Z$, replacing $\mathcal{L}_Z$ by $\mathcal{L}_Z^{\vee}$ in each of the other clauses, and making $p \models \bigvee \sigma_i$ iff $\exists i \in I$: $p\models \sigma_i$, then we have $\mathcal{L}_Z^{\vee}\sim\mathcal{L}_Z$.
\end{Proposition}
\begin{proof}
It is interesting to observe that even if the result is valid for all the semantics, the reason behind is not the same as in the case of bisimulation. In that case, we only need to apply the De Morgan laws to get the ``definition'' of $\vee$ as a combination of $\neg$ and  $\wedge$. However, for the rest of the semantics, we do not have negation as ``constructor'', but $\vee$ distributes over $\wedge$ and the prefix operator (because $\bigvee a \varphi_i= a \bigvee \varphi_i$), while negation is never applied to a formula $\varphi^{\prime} \in \mathcal{L}_Z^{\vee}$. Therefore, by floating away any $\vee$ in a formula in $\mathcal{L}_Z^{\vee}$, it becomes equivalent to a disjunction of formulas within the corresponding language $\mathcal{L}_Z$, and then the equivalence of both logics follows.
\end{proof}

\begin{remark}
Since we have $\perp = \neg \top = \neg \bigwedge_{i \in \emptyset} = \bigvee_{i \in \emptyset}$ , we conclude that $\perp \in \mathcal{L}^{\vee}_{Z}$ , and therefore all the logical semantics defined by these logics remain the same if we add $\perp$ and disjunction to their definitions. 
Moreover, $\wedge$ cannot be filtered by the prefix operator. By the way, this makes the difference between linear semantics (whose logics do not allow an arbitrary use of conjunction) and branching semantics (where we can arbitrarily use conjunction). It is important to note that $a\!\!\perp \sim \perp$ and therefore $a\!\! \perp \nsim \neg a \top$.
\end{remark}

\subsection{Observational characterizations for process semantics}
There is a clear connection between the observational and the logical semantics. In fact, we expected that once we had a unified presentation of the observational semantics it would be easy to transmute it into a unified presentation of the logical semantics. This was not that easy at the end, but certainly our unified logics were inspired by the previously obtained unified observational semantics. Moreover, we need these definitions if we want to check that our new logical semantics are indeed equivalent characterizations of the same semantics. Obviously, for the cases of the semantics in the classic spectrum we could instead compare (one by one) our new logics and those provided by van Glabbeek in \cite{Gla01}, but this cannot be done for any of the new semantics. Therefore, we briefly present  next the definitions (from  \cite{fgp09_observational}) needed to get these observational characterizations. 

One important fact about these characterizations is its finite character. All the considered observations are (structurally) finite, and this means that the characterizations work as long as we keep ourselves to the continuous side of the range of possible semantic domains. Therefore, we have to restrict ourselves to finite processes, or  at least to image-finite processes. It is for this class of processes that Th. \ref{teorema_preliminares} works.

\begin{definition}\label{obs:locales}
The sets \textit{$L_N$} of local observations corresponding to each of the \textit{N}-constrained simulations in the spectrum, and $L_N(p)$ of observations associated to a process \textit{p}, are defined as follows:
\begin{itemize} \label{def L_N}
\item S: $L_U=\{\cdot\}$, $L_U(p)=\cdot$.
\item CS: $L_C=Bool$, $L_C(p)$ is true if $p\models\textbf{0}$ and false otherwise.
\item RS: $L_I=\mathcal{P}(Act)$, $L_I(p)=I(p)=\{a |$ $a\in Act$ and $p \stackrel{a}{\rightarrow} \}$.
\item TS: $L_T=\mathcal{P}(Act^*)$, $L_T(p)$ is \textit{T(p)}, the set of traces of p.
\item 2S: $L_S=\{\|p\|_S \}$, $L_S(p)= \| p \|_S$ where $\| p \|_{S}$ denotes the simulation equivalence class of \textit{p}.
\item kS: $L_S=\{\|p\|_{(k-1)S} \}$, $L_S(p)= \| p \|_{(k-1)S}$, where $\| p \|_{kS}$ denotes the k-nested simulation equivalence class of \textit{p}.
\end{itemize}
\end{definition}

Each $N\in\{U,C,I,T,S\}$ induces uniformily an equivalence relation, that by abuse of notation we will also denote by N: $pNq ::= L_N(p)=L_N(q)$.

\begin{remark}
In the definition above we have considered both the trace semantics and the simulation semantics when defining $L_T$ and $L_S$. Certainly, we expect that the reader will be familiarized with these two classic semantics, and this is why we avoid a reminder of their definitions here. Also, there is another (more formal) reason for which we do this: the trace and the simulation semantics are two of the semantics to be classified by our systematic approaches, and it would not be nice to have their definitions in advance. Instead, we can apply (when needed) our definitions in a sliced way: based on \textit{U} we define plain simulations, and then the trace semantics, and once this is done, we have T and S to define TS and 2S. The same is valid, step by step, for all the nested simulation semantics.
\end{remark} 

\begin{definition} \label{def:ramificadas}
\begin{enumerate}
\item A branching general observation (bgo for short) of a process is a finite, non-empty tree whose arcs are labeled with actions in Act and whose nodes are labeled with local observations from $L_N$, for \textit{N} a constraint; the corresponding set $BGO_N$ is recursively defined as: $\langle l, \emptyset \rangle \in BGO_N$ for $l \in L_N$; $\langle l,\{(a_i,bgo_i) \mid i \in 1..n\}\rangle \in BGO_N$ for every $n \in \mathbb{N}, a_i \in Act$ and $bgo_i \in BGO_N$.
\item The set $BGO_N(p)$ of bgo's of a process \textit{p} corresponding to the constraint \textit{N} is $BGO_N(p)=\{\langle L_N(p),S\rangle$ $\mid S\subseteq \{(a,bgo) | bgo \in BGO_N(p^\prime), p \stackrel{a}{\rightarrow}p^\prime\}\}.$ We write $p\leq_N^bq$ if $BGO_N(p)\subseteq BGO_N(q)$.
\end{enumerate}
\end{definition}

\begin{theorem}[\cite{fgp09_observational}] \label{teorema_preliminares}
 For all $N\in\{U,C,I,T,S\}$ and any two processes \textit{p} and \textit{q}, \ $p\sqsubseteq_{NS}q$ iff $p\leq_N^b q\,$.
\end{theorem}

\begin{definition}\label{def:lineales}
\begin{enumerate}
\item The set $LGO_N$ of linear general observations (lgo for short) for the set of local observations $L_N$ is the subset of $BGO_N$ defined as: $\langle l, \emptyset\rangle \in LGO_N$ for each $l \in L_N$;  $\langle l,\{(a,lgo)\}\rangle$ whenever $a \in Act$ and $lgo \in LGO_N$.
\item The set $LGO_N(p)$ of lgo's of a process \textit{p} with respect to the set of local observations $L_N$ is $LGO_N(p)=BGO_N(p)\cap LGO_N.$
\end{enumerate}
\end{definition}

\begin{definition} \label{ordenes}
For $\zeta,\zeta^\prime\subseteq LGO_N$, we define the orders $\leq_N^{l}$, $\leq_N^{l\supseteq}$, $\leq_N^{lf}$, and $\leq_N^{lf\supseteq}$ by:
\begin{itemize}
\item $\zeta\leq_N^{l}\zeta^\prime$ $\stackrel{def}{\Leftrightarrow}$ $\zeta\subseteq \zeta^\prime$.
\item $\zeta\leq_N^{l\supseteq}\zeta^\prime$ $\stackrel{def}{\Leftrightarrow}$ $\forall$ $X_0a_1X_1\ldots X_n \in \zeta$ $\exists$ $Y_0a_1Y_1\ldots Y_n \in \zeta^\prime$ $\forall i \in 0..n$ $X_i\supseteq Y_i$.
\item $\zeta\leq_N^{lf}\zeta^\prime$ $\stackrel{def}{\Leftrightarrow}$ $\forall$ $X_0a_1X_1\ldots X_n \in \zeta$ $\exists$ $Y_0a_1Y_1\ldots Y_n \in \zeta^\prime$ $X_n=Y_n$.
\item $\zeta\leq_N^{lf\supseteq}\zeta^\prime$ $\stackrel{def}{\Leftrightarrow}$ $\forall$ $X_0a_1X_1\ldots X_n \in \zeta$ $\exists$ $Y_0a_1Y_1\ldots Y_n \in \zeta^\prime$ $X_n\supseteq Y_n$.
\end{itemize}
\end{definition}

\begin{definition}
Given two processes p and q and $Z \in \{ l, l\supseteq, lf, lf\supseteq\}$, we write $p\leq_N^Zq$ \ iff \ $LGO_N(p)\leq_N^Z LGO_N(q)$. We will denote the corresponding equivalence by $=_N^Z$.
\end{definition}

In the cases in which there is no previously known (equivalent) definition for our new semantics, the definition above will give us ``the'' definition of each one of these new semantics; instead, each of the linear semantics in the old spectrum has a companion in our \emph{enlarged spectrum}. For instance, the linear semantics in the diamond to the right of RS (see Figure \ref{spectrum_fgp}) satisfy the following theorem.

\begin{theorem} \label{ordenes_semanticas}
(1) $p\sqsubseteq_{RT}q$ iff $p\leq_I^lq$; \ \ (2) $p\sqsubseteq_{FT}q$ iff $p\leq_I^{l\supseteq}q$; \ \ (3) $p\sqsubseteq_{R}q$ iff $p\leq_I^{lf}q$;  \ \ (4) $p\sqsubseteq_{F}q$ iff $p\leq_I^{lf\supseteq}q$.
\end{theorem}
\section{A new logical characterization of the most popular semantics} \label{logical_characterizations_popular_semantics}

Next we will present in a uniform way the new logics that characterize the different semantics. Each of them is defined by a set of rules, and as usual we assume that only the formulas that can be obtained by finite application of these rules are in the defined logics. We begin by studying the particular cases of the best known classical semantics, that is, those at the layer of Ready Simulation in the \emph{enlarged spectrum}. All of them use in some way the set of formulas $\mathcal{L}_I=\{a\top \mid a \in Act\}$ that characterizes the initial offers of a process. In Section \ref{logical_characterizations_coarsest_semantics}, we will present the logics for the rest of the semantics in a unified way. 

\begin{definition}\label{def_logica}
\textbf{Ready Simulation semantics (RS)}: we define the set of formulas $\mathcal{L}^{\prime}_{RS}$ for ready simulation semantics by \bran{\conjsimold{RS}{I}}{RS}. \\
\textbf{Ready traces semantics (RT)}: we define the set of formulas $\mathcal{L}^{\prime}_{RT}$ for ready trace semantics by \lin{\conjyformparticular{RT}{I}}{RT} \\
\textbf{Failure traces semantics (FT)}: we define the set of formulas $\mathcal{L}^{\prime}_{FT}$ for failure traces semantics by \linfparticular{I}{FT} \\
\textbf{Readiness semantics (R)}: we define the set of formulas $\mathcal{L}^{\prime}_{R}$ for readiness semantics by \lincparticular{I}{R} \\
\textbf{Failures semantics (F)}: we define the set of formulas $\mathcal{L}^{\prime}_{F}$ for failures semantics by \linfcparticular{I}{F}
\end{definition}

One can immediately check in the definition above that $\mathcal{L}^{\prime}_{RS}\subseteq \mathcal{L}_B$, thus obtaining that Ready Simulation semantics is coarser than Bisimulation equivalence. We also have $\mathcal{L}^{\prime}_{F}\subseteq \mathcal{L}^{\prime}_{R}$, $\mathcal{L}^{\prime}_{F}\subseteq \mathcal{L}^{\prime}_{FT}$, $\mathcal{L}^{\prime}_{R}\subseteq \mathcal{L}^{\prime}_{RT}$, $\mathcal{L}^{\prime}_{FT}\subseteq \mathcal{L}^{\prime}_{RT}$ and $\mathcal{L}^{\prime}_{RT}\subseteq \mathcal{L}^{\prime}_{RS}$, which can be interpreted in the same way. Let us now focus our attention on the third rule of the definition of $\mathcal{L}^{\prime}_{RS}$: the unrestricted use of conjunction corresponds to the branched character of the semantics. Moreover, the two first rules allow us to fix the set of offers at the states of the process as \textit{I}-simulations impose. Instead, the linear semantics only allow the use of conjunction to join the simple formulas that permit us to fix the set of offers along a computation in the case of the readies-based semantics, or their over-approximations (obtained by means of the negated formulas $\neg a\top$), in the case of the failures-based semantics. Finally, notice how these simple formulas can only be checked at the end, for the simpler coarser semantics.

Now, for $X \in \{RS, RT, FT, R, F\}$ we can prove that each of the logics, $\mathcal{L}^{\prime}_{X}$, is a superset of the corresponding logic, $\mathcal{L}_{X}$, defined by van Glabbeek in \cite{Gla01}. To be precise, for the cases of FT and F semantics we need to remove the syntactic sugar used by van Glabbeek.

\begin{Proposition}
\begin{enumerate}
\item $\mathcal{L}^{\prime}_{RS}\supseteq \mathcal{L}_{RS}$. We also have $\mathcal{L}_{RS}\varsubsetneq \mathcal{L}^{\prime}_{RS}\,$.
\item $\mathcal{L}^{\prime}_{RT}\supseteq \mathcal{L}_{RT}$. We also have $\mathcal{L}_{RT}\subsetneq \mathcal{L}^{\prime}_{RT}\,$. 
\item $\mathcal{L}^{\prime}_{FT}\supseteq$ desugared$(\mathcal{L}_{FT})$, where the desugaring function removes the syntactic sugar used in $\mathcal{L}_{FT}\,$.
\item $\mathcal{L}^{\prime}_{R}\supseteq \mathcal{L}_{R}$. We also have $\mathcal{L}_{R}\varsubsetneq \mathcal{L}^{\prime}_{R}\,$.
\item $\mathcal{L}^{\prime}_{F}\supseteq$ desugared($\mathcal{L}_{F}$), where the desugared function removes the syntactic sugar used in $\mathcal{L}_F\,$.
\end{enumerate}
\end{Proposition}
\begin{proof}
All of them are simple and similar, so we will only present the proof of 2.
\begin{itemize}
\item$\underline{2|}$ To prove that $\mathcal{L}^{\prime}_{RT} \supseteq \mathcal{L}_{RT}$ it is sufficient to show that for every $X \subseteq Act$ and any $\varphi \in \mathcal{L}_{RT}$, the formula $(\bigwedge_{af \in X} a\top \wedge \bigwedge_{b \notin X} \neg b\top) \wedge \varphi$ belongs to $\mathcal{L}^{\prime}_{RT}$. Note that $b \notin X$ is equivalent to $b \in \overline{X}$, so taking $X_1=X$ and $X_2=\overline{X}$ we have that the considered formula belongs to $\mathcal{L}^{\prime}_{RT}$. To prove that  $\mathcal{L}_{RT} \subset \mathcal{L}^{\prime}_{RT}$, it is sufficient to note that $(\neg b\top) \wedge \varphi$ belongs to $\mathcal{L}^{\prime}_{RT}\,$, by simply taking $X_1=\emptyset$ and $X_2=\{b\}$, but it does not belong to $\mathcal{L}_{RS}$.
\end{itemize}    
\end{proof}

We have said in our Introduction that our logics are chosen as large as necessary, to obtain more natural characterizations. This is why, in most of the cases, we have obtained a logic larger than that proposed by van Glabbeek. In order to prove the equivalences between ours and van Glabbeek's logics, we have to show that the new formulas that we included in our logics are in fact redundant.
\begin{Proposition} \label{equiv_entre_logicas}
We have (1) $\mathcal{L}_{RS}\sim \mathcal{L}^{\prime}_{RS}$; (2) $\mathcal{L}_{RT}\sim \mathcal{L}^{\prime}_{RT}$;  (3) $\mathcal{L}_{FT}\sim \mathcal{L}^{\prime}_{FT}$; (4) $\mathcal{L}_{R}\sim \mathcal{L}^{\prime}_{R}$ and (5) $\mathcal{L}_{F}\sim \mathcal{L}^{\prime}_{F}$.
\end{Proposition}
\begin{proof}
As above we will only present one of the proofs.
\begin{itemize}
\item$\underline{2|}$ We have seen that the formulas in $\mathcal{L}_{RT}$ are particular cases of the formulas in $\mathcal{L}^{\prime}_{RT}$, those that totally define the offers at the states along a computation (when we apply the second clause in the definition of $\mathcal{L}^{\prime}_{RT}$ taking $X_2 = \overline{X_1}$).
Instead, our more general formulas $(\bigwedge_{a\top \in X_1} a\top \wedge \bigwedge_{b\top \in X_2}\neg b\top) \wedge \varphi$, where $\varphi \in \mathcal{L}^{\prime}_{RT}$, could give us some partial information, combining both positive information $a\top \in X_1$ and negative information $b\top \in X_2$, which tells us that we are in an arbitrary state $X$, satisfying $X_1 \subseteq X \subseteq \overline{X_2}$. But we can replace these formulas by the disjunction of all the formulas describing any of these possible offers $X$. By repeating this procedure at each level of the formula, we finally obtain a disjunction of formulas in $\mathcal{L}_{RT}$. To conclude, it is enough to apply Prop. \ref{disyuncion_teo}.
\end{itemize}
\end{proof}

In the following, when we consider a logic $\mathcal{L}_{Z}$ and the index $Z$ refers to some concrete semantics, as is the case with $RS$, $RT$, $FT$, $R$, $F$ above, by abuse of notation we will simply write $\sqsubseteq^{\prime}_{Z}$ instead of $\sqsubseteq_{\mathcal{L}^{\prime}_{Z}}$ when referring to the preorder induced by the logic $\mathcal{L}^{\prime}_Z$.

\begin{theorem}\label{teoequiv}
\begin{enumerate}
\item The logical semantics $\sqsubseteq^{\prime}_{RS}$ induced by the logic $\mathcal{L}^{\prime}_{RS}$ is equivalent to the observational branching semantics defined by $\leq_{I}^{b}$, generated by the set of branching general observations $BGO_I$.
\item The logical semantics  $\sqsubseteq^{\prime}_{RT}$ (resp. $\sqsubseteq^{\prime}_{FT}$, $\sqsubseteq^{\prime}_{R}$,  $\sqsubseteq^{\prime}_{F}$) induced by the logic $\mathcal{L}^{\prime}_{RT}$ (resp. $\mathcal{L}^{\prime}_{FT}$, $\mathcal{L}^{\prime}_{R}$, $\mathcal{L}^{\prime}_{F}$) is equivalent to the observational linear semantics defined by the domain of linear general observations $LGO_I$, ordered by $\leq_I^{l}$ (resp.  $\leq_I^{l\supseteq}$, $\leq_I^{lf}$, $\leq_I^{lf\supseteq}$,) defined at Def. \ref{ordenes}.
\end{enumerate}
\end{theorem}
\begin{proof}
It is a consequence of  Prop. \ref{equiv_entre_logicas}, the results by van Glabbeek collected in Table \ref{logic_table}, Th. \ref{teorema_preliminares} and Th. \ref{ordenes_semanticas}.
\end{proof}

\begin{figure}[h]
\scalebox{0.47}{\huge
$\begin{array}{c c@{\hspace{3.25ex}\vline\hspace{3.25ex}} c c}
\textit{P}_1&\textit{P}_2 &\textit{P}_5 &\textit{P}_6\\ &&&\\
\xymatrix{
& &\bullet\ar[dl]_a \ar[d]_a\ar[drr]^a& & &\\
&\bullet\ar[d]_b  &\bullet & & \bullet \ar[dl]_b \ar[d]_c\ar[dr]^e &\\
& \bullet & &\bullet \ar[d]_d &\bullet &\bullet\\
& & & \bullet & &}
&
\xymatrix{
& & &\bullet\ar[dll]_a \ar[d]_a\ar[drr]^a& & &\\
&\bullet\ar[d]_b & &\bullet \ar[dl]_b \ar[d]_c & & \bullet \ar[dl]_b \ar[d]_c\ar[dr]^e &\\
&\bullet & \bullet \ar[d]_d&\bullet &\bullet \ar[d]_d &\bullet &\bullet\\
& & \bullet & & \bullet & &}
&
\xymatrix{
&\bullet \ar[dl]_a\ar[d]_a\ar[dr]^a&\\ 
\bullet \ar[d]_b& \bullet \ar[d]_b& \bullet \ar[d]^b\\ 
\bullet \ar[d]_c&\bullet \ar[dl]_c \ar[dr]^d& \bullet \ar[d]^d\\ 
\bullet &&\bullet }
&
\xymatrix{
&\bullet\ar[dl]_a\ar[dr]^a &\\
\bullet\ar[d]_b& & \bullet\ar[d]^b\\
\bullet\ar[d]_c& & \bullet \ar[d]^d\\
\bullet&&\bullet}
\\
&&&
\\
\textit{P}_3&\textit{P}_4 &\textit{P}_7 &\textit{P}_8 \\&&& \\
\xymatrix{
& & &\bullet\ar[dl]_a \ar[d]_a\ar[drr]^a& & &\\
& &\bullet\ar[d]_b &\bullet \ar[d]_b & & \bullet \ar[dl]_b \ar[d]_c\ar[dr]^e &\\
& & \bullet & \bullet \ar[d]_d &\bullet \ar[d]_d &\bullet &\bullet\\
& & &\bullet & \bullet & &}
&
\xymatrix{
& &\bullet\ar[dll]_a \ar[dr]^a&&\\
\bullet\ar[d]_b& & & \bullet \ar[dl]_b \ar[d]_c\ar[dr]^e &\\
\bullet & & \bullet  \ar[d]_d&\bullet & \bullet \\
& &  \bullet  & &}
&
\xymatrix{
&\bullet\ar[d]_a&\\
&\bullet\ar[dl]_b \ar[dr]^b&\\
\bullet\ar[d]_c& & \bullet \ar[d]^d\\
\bullet&&\bullet}
&
\xymatrix{
& &\bullet\ar[dl]_a\ar[dr]^a& &\\
&\bullet\ar[dl]_b \ar[dr]^b& &\bullet \ar[d]_b&\\
\bullet\ar[d]_c& & \bullet \ar[d]^d & \bullet \ar[d]_c&\\
\bullet&&\bullet &\bullet &}
\end{array}$}
\caption{Example to show the strength of the different logics}\label{example1} 
\end{figure}
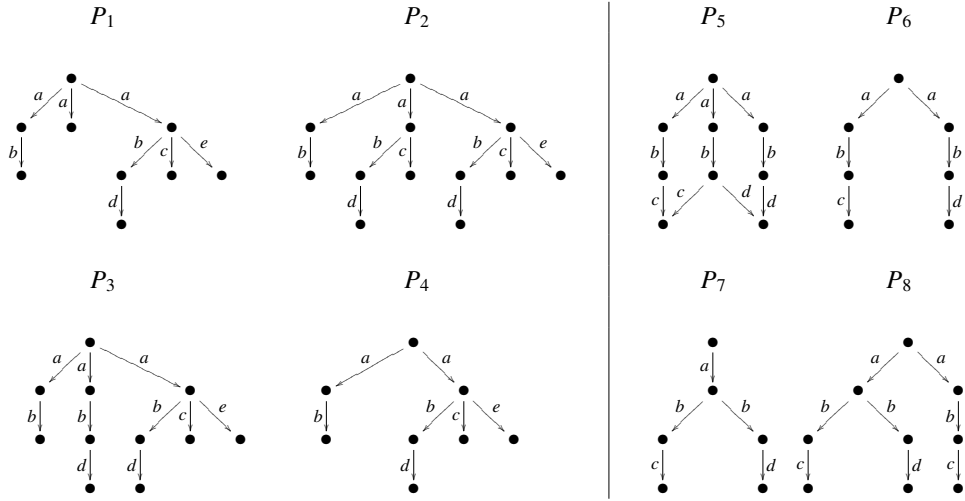
\normalsize
\begin{example}
Figure \ref{example1} shows a collection of  examples to illustrate the differences between the semantics in the layer of RS at the spectrum. All the stated equivalences can be checked by taking any arbitrary formula from the logic defining each of the semantics.
For readability, we omit the last $\top$ in all subformulas. Besides, $\sim_{X}$, (resp. $\nsim_{X}$) , where \emph{X} is a set of indexes, represents any $\sim_{Z}$ ( resp. $\nsim_{Z}$), with  $Z \in\emph{X}$. 
\begin{itemize}
\item $P_1\nsqsubseteq^{\prime}_{F}\!P_2$, and then $P_1\nsqsubseteq^{\prime}_{\{R,\;FT,\;RT,\;RS\}}\!P_2$; this is because $P_1\models a(\neg b \wedge \neg c)$, but $P_2$ does not.
\item $P_2\sim_{F}P_3$, but $P_2\nsqsubseteq^{\prime}_{\{R,\;FT\}}P_3$ and then $P_2\nsqsubseteq^{\prime}_{\{RT,\;RS\}}P_3$, using that $P_2\models a(\neg e\wedge c)$, but $P_3$ does not.
\item $P_3\!\!\sim_{\{F,\;R\}}\! \!P_4$, but $P_3\!\nsqsubseteq^{\prime}_{FT}\!\!P_4$ and then $P_3\!\nsqsubseteq^{\prime}_{\{RT,\;RS\}}\!\!P_4$, because $P_3\! \!\models \!a(\neg c \wedge b(\neg e \wedge d))$, but $P_4$ does not.
\item $P_5\sim_{\{F,\;FT\}}P_6$, but $P_5\nsqsubseteq^{\prime}_{R}P_6$ and then $P_5\nsqsubseteq^{\prime}_{\{RT,\;RS\}}P_6$, using that $P_5\models ab(c\wedge d)$, but $P_6$ does not.
\item $P_6\sim_{\{F,\;R,\;RT,\;FT\}}P_7$, but $P_7\nsqsubseteq^{\prime}_{RS}P_6$, using that $P_7\models a(bc\wedge bd)$, but $P_6$ does not.
\item $P_7\sim_{\{F,\;R,\;RT,\;FT,\;RS\}}P_8$.
\end{itemize}
\end{example}
\section{Our new unified logical characterizations of the semantics}\label{logical_characterizations_coarsest_semantics}
Inspired by the semantics studied in Section \ref{logical_characterizations_popular_semantics}, next we define the general format for the logics characterizing each of the semantics in the \emph{enlarged spectrum}. We start by enlarging the spectrum a bit more, to include all the elements needed to characterize the rest of the semantics in a systematic way.

\begin{definition}
\begin{enumerate}
\item \textbf{Universal semantics (U)}: We define the set of Universal formulas, $\mathcal{L}^{\prime}_{U}$, that characterizes the trivial semantics that identifies all the processes, by $\mathcal{L}^{\prime}_{U} = \{ \top \}$.
\vspace{-0.15cm}
\item \textbf{Complete semantics (C)}: It is defined by $\sqsubseteq_C$, taking $p\sqsubseteq_C q$ ::= ($p \stackrel{a}{\rightarrow}$ $\Rightarrow$ $\exists$ b $\in Act$ $q \stackrel{b}{\rightarrow}$). That is, it only distinguishes terminated processes (equivalent to \textbf{0}) from non-terminated ones. We define the set of Complete formulas $\mathcal{L}^{\prime}_{C}$ characterizing it, by $\mathcal{L}^{\prime}_{C} = \{ \top, \neg0\}$.
\item \textbf{Initial offer semantics (I)}: It is defined by $\sqsubseteq_I$, taking $p \sqsubseteq_I q$ ::= $I(p)\subseteq I(q)$. That is, it only observes the set of initial actions of a process, $I(p)=\{a \mid a\in Act \wedge p \stackrel{a}{\rightarrow}\}$. We define the set of Initial offer formulas $\mathcal{L}^{\prime}_{I}$ characterizing it, by $\mathcal{L}^{\prime}_{I} = \{ \top, \neg0 \}$ $\cup$ $\{a\top \mid a \in \emph{Act}\}$.
\end{enumerate}
\end{definition}

In the definition above the sub-formula $\neg 0$ is just syntactic sugar for the formula $\neg(\bigwedge_{a \in Act}\neg a \top)$. Therefore, all these new logics are indeed sublogics of $\mathcal{L}_{HM}$, and we do not need to define their semantics.

Note that $\mathcal{L}^{\prime}_I$ is a bit larger than the logic $\mathcal{L}_I$ used in Section \ref{logical_characterizations_popular_semantics}. Once again, this is so in order to get a more uniform presentation of our logics: $\neg 0$ is indeed redundant. As a consequence, we immediately obtain that the Complete semantics is coarser than the Initial offer semantics, because $\mathcal{L}^{\prime}_{C}$ $\subseteq$ $\mathcal{L}^{\prime}_{I}$. Based on this result we will also easily obtain that the Complete Simulation is coarser than the Ready Simulation.

\subsection{The simulation semantics}
As discussed in \cite{fgp09_observational}, the simulation semantics constitute the spine of the new spectrum. Moreover, all of them are defined in a homogeneous way using the notion of constrained simulation from \cite{fg08}.

\begin{definition}\label{logica_sim}
Given a set of formulas $\mathcal{L}^{\prime}_{N}$ defining a semantics N, we define the set of formulas $\mathcal{L}^{\prime}_{NS}$ that defines the \textit{N}-constrained simulation semantics by \brannew{\conjsim{NS}{N}}{NS}
\end{definition}

Taking $N\in\{U, C, I\}$ we obtain $\mathcal{L}^{\prime}_{US}$, $\mathcal{L}^{\prime}_{CS}$ and $\mathcal{L}^{\prime}_{IS}$, that in the first and last cases we rewrite as $\mathcal{L}^{\prime}_{S}$ and $\mathcal{L}^{\prime}_{RS}$, respectively, in order to emphasize the classic notation for simulation semantics. Once that we have $\mathcal{L}^{\prime}_{S}$ we can also obtain $\mathcal{L}^{\prime}_{SS}$, that we will also denote as $\mathcal{L}^{\prime}_{2S}$. To complete the collection of simulation semantics we will only need $\mathcal{L}^{\prime}_{TS}$, that will be based on $\mathcal{L}^{\prime}_{T}$, to be defined in the next section.

If we compare the definition above with the particular case of Ready Simulation in Def. \hspace{-0.07cm}\ref{def_logica}, the differences concern the two first rules, by means of which we impose that the process will traverse states which are in the corresponding N-equivalence class all along the tree of computations checked by a formula in $\mathcal{L}^{\prime}_{NS}$. Next we state the equivalence between our logics and those by van Glabbeek in \cite{Gla01}. 

\begin{Proposition}
We have (1) $\mathcal{L}^{\prime}_{S} \sim \mathcal{L}_{S}$, (2) $\mathcal{L}^{\prime}_{CS} \sim \mathcal{L}_{CS}$ and  (3) $\mathcal{L}^{\prime}_{2S} \sim \mathcal{L}_{2S}$.
\end{Proposition}

\subsection{Logical characterization of the linear semantics} \label{logica_lineales}
We start by defining the closure operators, by means of which we are able to express to which extent conjunction and negation can be used at the logical characterizations of each of the linear semantics.

\begin{definition}
Given a logical set $\mathcal{L}^{\prime}_{N}$ with $N \in \{U,C,I,T,S\}$, we define:
\begin{enumerate}
\item \textbf{Its symmetric closure} $\mathcal{L}_{N}^{\equiv}$ by: $\sigma \in \mathcal{L}^{\prime}_{N} \Rightarrow \sigma \in \mathcal{L}_{N}^{\equiv}$ and $\neg \sigma \in \mathcal{L}_{N}^{\equiv}$; $\sigma_i \in \mathcal{L}_{N}^{\equiv} \hspace{0.1cm} \forall i\in I \Rightarrow \bigwedge_{i\in I}\sigma_i \in \mathcal{L}_{N}^{\equiv}$.
\item \textbf{Its negative closure}  $\mathcal{L}_{N}^{\neg}$ by: $\sigma \in \mathcal{L}^{\prime}_{N} \Rightarrow \neg \sigma \in \mathcal{L}_{N}^{\neg}$; $\sigma_i \in \mathcal{L}_{N}^{\neg} \hspace{0.1cm} \forall i\in I \Rightarrow \bigwedge_{i\in I}\sigma_i \in \mathcal{L}_{N}^{\neg}$.
\item \textbf{Its positive closure} $\mathcal{L}_{N}^{\surd}$ by: $\sigma \in \mathcal{L}^{\prime}_{N} \Rightarrow \sigma \in \mathcal{L}_{N}^{\surd}$; $\sigma_i \in \mathcal{L}_{N}^{\surd} \hspace{0.1cm} \forall i\in I \Rightarrow \bigwedge_{i\in I}\sigma_i \in \mathcal{L}_{N}^{\surd}$.
\end{enumerate}
\end{definition}

Whenever we have a bag of ``good'' properties (such as $\mathcal{L}^{\prime}_{N}$ above), if we want to assert by means of a single formula which is the subset of properties that a certain element satisfies, it is not sufficient to assert that it satisfies each one of them: we also need to assert that it does not satisfy all the rest. This is why we need formulas in the symmetric closure. Instead, if we can only manage formulas from the negative (resp. positive) closure, we can only assert that the element has at most (resp. at least) the enumerated properties. Next we present the unified logics for all the linear semantics in the \emph{enlarged spectrum}.

\begin{definition}\label{linear_logica}
Inspired by the orders $\leq_N^{l}$, $\leq_N^{l\supseteq}$, $\leq_N^{lf}$ and $\leq_N^{lf\supseteq}$, we define the set of formulas $\mathcal{L}^{\prime}_{\leq_N^{l}}$, $\mathcal{L}^{\prime}_{\leq_N^{l\supseteq}}$, $\mathcal{L}^{\prime}_{\leq_N^{lf}}$ and $\mathcal{L}^{\prime}_{\leq_N^{lf\supseteq}}$, respectively,  by means of the rules:
\begin{enumerate}
\item \lin{\conjyform{{\leq_N^{l}}}{N}}{{\leq_N^{l}}}
\item \linf{N}{{\leq_N^{l\supseteq}}}
\item \linc{N}{{\leq_N^{lf}}}
\item \linfc{N}{{\leq_N^{lf\supseteq}}}
\end{enumerate}
\end{definition}

Note that for the coarsest semantics (e.g. those corresponding to plain refusals and plain readiness when $N=I$) we only observe $N$ at the end of the formula. Instead, the other two logics introduce additional conjunctions that allow to observe $N$ along the computations. Moreover, we have used the negative (resp. symmetric) closure in the ``failures based''  (resp. ``readies based'') semantics. 

We can use the positive closure to define two new semantics that were not studied in \cite{fgp09_equational,fgp09_observational} nor elsewhere, as far as we know. They are defined by observing partial offers along a computation, or just at its end. We say that \emph{X} is a partial offer of \emph{p} if $X \subseteq I(p)$. It is clear the duality w.r.t. the failures semantics, where \emph{F} is a failure of \emph{p} if $I(p) \subseteq \overline{F}$. We can introduce these two new semantics at each layer of the spectrum, by defining the corresponding partial offers for each $N \in \{U, C, I, T, S\}$.

\begin{definition}
\begin{enumerate}
\item The semantics of \textbf{partial offer traces} for the constraint N is that defined by the logic $\mathcal{L}^{\prime}_{\leq_N^{l\subseteq}}$ with \linfmo{N}{\leq_N^{l\subseteq}}
\item The semantics of \textbf{partial offers} for the constraint N is that defined by the logic $\mathcal{L}^{\prime}_{\leq_N^{lf\subseteq}}$ with \linfcmo{N}{{\leq_N^{lf\subseteq}}}
\end{enumerate}
\end{definition}

Duality between failures and partial offers causes the picture of the complete layer of linear semantics for each \textit{N} to become two diamonds that share the side corresponding to the readies-based semantics.

\begin{Proposition}
\begin{enumerate}
\item $\mathcal{L}^{\prime}_{F}$ and $\mathcal{L}^{\prime}_{\leq_I^{lf\subseteq}}$ are not comparable: $p\leq_I^{lf\supseteq}\!q$ $\nRightarrow$ $p\leq_I^{lf\subseteq}\!q\ $ and $\ p\leq_I^{lf\subseteq} \!q$ $\nRightarrow$ $p\leq_I^{lf\supseteq}\! q$.
\item $\mathcal{L}^{\prime}_{FT}$ and $\mathcal{L}^{\prime}_{\leq_I^{l\subseteq}}$ are incomparable: $p\leq_I^{l\supseteq}q$ $\nRightarrow$ $p\leq_I^{l\subseteq}q$ and $p\leq_I^{l\subseteq}q$ $\nRightarrow$ $p\leq_I^{l\supseteq}q$.
\end{enumerate}
\end{Proposition}
\begin{proof}
In fact we have a stronger result combining the two statements: if we consider $p=ab+ac$, $q=a(b+c)$ and $r=p+q$, we have that $\ p\sim_{l\supseteq}r\ $ but $\ r\nleq_I^{lf\subseteq}p\ $ and $\ q\sim_{l\subseteq}r\ $ but $\ r\nleq_I^{lf\supseteq}q$. 
\end{proof}

We could obtain similar counterexamples for $N \in \{T,S\}$. Instead, for $N \in \{U,C\}$, which produce the trace semantics and the complete traces semantics, respectively, it is easy  to prove that the six logics of the layer are indeed equivalent.

\begin{Proposition}
We have (1) $\mathcal{L}^{\prime}_{\leq^{lf}_U}$ = $\mathcal{L}^{\prime}_{\leq^{l}_U}$ = $\mathcal{L}^{\prime}_{\leq^{l\supseteq}_U}$ = $\mathcal{L}^{\prime}_{\leq^{l\subseteq}_U}$ = $\mathcal{L}^{\prime}_{\leq^{lf\supseteq}_U}$ = $\mathcal{L}^{\prime}_{\leq^{lf\subseteq}_U}$ = $\mathcal{L}_{T}$ and (2) $\mathcal{L}^{\prime}_{\leq^{lf\supseteq}_C}$ = $\mathcal{L}^{\prime}_{\leq^{lf\subseteq}_C}$ = $\mathcal{L}^{\prime}_{\leq^{l\supseteq}_C}$ = $\mathcal{L}^{\prime}_{\leq^{l\subseteq}_C}$ = $\mathcal{L}^{\prime}_{\leq^{lf}_C}$ = $\mathcal{L}^{\prime}_{\leq^{l}_C}$ = $\mathcal{L}_{CT}$.
\end{Proposition}

An interesting result illustrating the genericity of our characterizations concerns one of the finest semantics in the classic spectrum: Possible Future (PF). We find PF in Figure \ref{spectrum_Gla} below 2S, probably because the more accurate simulation semantics TS was not (yet) included in the spectrum. This is corrected in the \emph{enlarged spectrum} in Figure \ref{spectrum_fgp}. Considering $N=T$, we have indeed the following result.

\begin{Proposition}
We have $\mathcal{L}^{\prime}_{\leq^{lf}_T}$ = $\mathcal{L}_{PF}$.
\end{Proposition}

\subsection{Logical characterization of the deterministic branching semantics} \label{logica_ramificada}
Next we consider the deterministic branching semantics. In the classic spectrum the only such semantics is \emph{Possible Worlds} (PW), but there is one such semantics for each level of the \emph{enlarged spectrum}.

\begin{definition}\label{logica_det}
For each $N \in \{U,C,I,T,S\}$, we define the formulas of $\mathcal{L}^{\prime}_{D_N}$ by: \dbran{\conjyform{D_N}{N}}{D_N}
\end{definition}

For $N=I$ we obtain the unified logical characterization of the PW semantics.

\begin{Proposition}
We have $\mathcal{L}^{\prime}_{D_I}\supseteq \mathcal{L}_{PW}$.
\end{Proposition}

By the way, $\mathcal{L}^{\prime}_{D_I}$ and $\mathcal{L}_{PW}$ are not equivalent, but this is caused by the fact that the original logical characterization $\mathcal{L}_{PW}$ was wrong. It can be checked, for instance, that taking $p=abc+a(bc+d)+ab$ and $q=a(bc+d)+ab$ we have $p \nsim_{PW} q$, but $p\sim_{\mathcal{L}_{PW}}q$, since $\mathcal{L}_{PW}$ cannot ``observe'' the intermediate offer that makes the possible world $abc$ different from those of \textit{q}. Instead, the formula $\varphi\equiv a(\neg d \wedge bc) \in \mathcal{L}_{D_I}$ is enough to distinguish \textit{p} and \textit{q}, since we have $p\models \varphi$ and $q \nvDash \varphi$.

\begin{table}[h]
\begin{center}
\scalebox{0.85}{\footnotesize
\begin{tabular}{|c|c|c|c|c|c||c|}
\hline
\backslashbox{Formulas}{Constraints ($\mathcal{N}$)}& U & C & I & T & S & B\\
\hline
$\top  \in \mathcal{L}^{\prime}_{\mathcal{N}}$ & $\bullet$ & $\bullet$ & $\bullet$ & $\bullet$ & $\nu$ & $\nu$\\
\hline
$\textbf{$\neg \top$ = $\perp$}\in \mathcal{L}^{\prime}_{\mathcal{N}}$& $\nu$ &$\nu$&$\nu$ & $\nu$ & $\nu$ & $\nu$\\
\hline
$\textbf{$\neg0$}\in \mathcal{L}^{\prime}_{\mathcal{N}}$& &$\bullet$ & $\bullet$ & $\nu$ & $\nu$ & $\nu$\\
\hline
$a \in Act \Rightarrow \hspace{0.075cm} a\top \in \mathcal{L}^{\prime}_{\mathcal{N}}$ & & &$\bullet$ &$\nu$ &$\nu$ &$\nu$\\
\hline
$\varphi \in \mathcal{L}^{\prime}_{\mathcal{N}}, \hspace{0.075cm} a \in Act \Rightarrow$ & & & &\multirow{2}{*}{$\bullet$} &\multirow{2}{*}{$\bullet$} &\multirow{2}{*}{$\bullet$} \\
$a \varphi \in \mathcal{L}^{\prime}_{\mathcal{N}}$& & & & & &\\
\hline
$\varphi_i \in \mathcal{L}^{\prime}_{\mathcal{N}} \hspace{0.075cm} \forall i \in I \Rightarrow$ & & & & &\multirow{2}{*}{$\bullet$} &\multirow{2}{*}{$\bullet$}\\
$\bigwedge_{i \in I} \varphi_i \in \mathcal{L}^{\prime}_{\mathcal{N}}$ & & & & & &\\
\hline
$\varphi \in \mathcal{L}^{\prime}_{\mathcal{N}} \Rightarrow$ & & & & & &\multirow{2}{*}{$\bullet$}\\
$\neg \varphi \in \mathcal{L}^{\prime}_{\mathcal{N}}$ & & & & & &\\
\hline
\end{tabular}}
\vspace{1ex}
\caption{Logical characterizations of the semantics used as constraints in the \emph{N}-constrained semantics} \label{our_constraint_table}
\vspace{-0.2cm}
\end{center}
\end{table}

\begin{table}[h]
\begin{center}
\scalebox{0.85}{\footnotesize
\begin{tabular}{|c|c|c|c|c|c|c||c}
\hline
\multirow{4}{*}{\backslashbox{Formulas}{Semantics ($\mathcal{Y_N}$)}}& \multirow{2}{*}{$\mathbf{\leq_\textit{N}^{lf\supseteq}}$} & \multirow{2}{*}{$\mathbf{\leq_\textit{N}^{lf}}$} & \multirow{2}{*}{$\mathbf{\leq_\textit{N}^{l\supseteq}}$} & \multirow{2}{*}{$\mathbf{\leq_\textit{N}^{l}}$} &\multirow{2}{*}{$D_N$} & \multirow{2}{*}{$NS$} & \multicolumn{1}{c|}{\multirow{2}{*}{$N \in \{U,C,I,T,S\}$}}\\ 
& & & & & & & \multicolumn{1}{c|}{}\\
\cline{2-8}
& \multirow{2}{*}{F} & \multirow{2}{*}{R} & \multirow{2}{*}{FT} & \multirow{2}{*}{RT} & \multirow{2}{*}{PW} & \multirow{2}{*}{RS} & \multicolumn{1}{c|}{\multirow{2}{*}{when $N=I$}}\\
& & & & & & & \multicolumn{1}{c|}{}\\
\hline
$\top  \in \mathcal{L}^{\prime}_{{\mathcal{Y_N}}}$ & $\bullet$  & $\bullet$ & $\bullet$ & $\bullet$ & $\bullet$ & $\nu$ &\\
\cline{1-7}
$\varphi \in \mathcal{L}^{\prime}_{{\mathcal{Y_N}}}, \hspace{0.075cm} a \in Act \Rightarrow$ & \multirow{2}{*}{$\bullet$}  & \multirow{2}{*}{$\bullet$} & \multirow{2}{*}{$\bullet$} & \multirow{2}{*}{$\bullet$} & \multirow{2}{*}{$\nu$} & \multirow{2}{*}{$\bullet$} &\\
$a \varphi \in \mathcal{L}^{\prime}_{{\mathcal{Y_N}}}$ &  &  &  &  &  &  &\\
\cline{1-7}
$\varphi \in \mathcal{L}_{N}^{\neg}\Rightarrow$ & \multirow{2}{*}{$\bullet$}  & \multirow{2}{*}{$\nu$} & \multirow{2}{*}{$\nu$}  & \multirow{2}{*}{$\nu$}  & \multirow{2}{*}{$\nu$}  & \multirow{2}{*}{$\nu$}  &\\
$\varphi \in \mathcal{L}^{\prime}_{\mathcal{Y_N}}$ &  &  &  &  &  &  &\\
\cline{1-7}
$\varphi \in \mathcal{L}_{N}^{\equiv}\Rightarrow$ &  & \multirow{2}{*}{$\bullet$}  &  & \multirow{2}{*}{$\nu$}  & \multirow{2}{*}{$\nu$}  & \multirow{2}{*}{$\nu$}  &\\
$\varphi \in \mathcal{L}^{\prime}_{\mathcal{Y_N}}$ &  &  &  &  &  &  &\\
\cline{1-7}
$\varphi \in \mathcal{L}^{\prime}_{{\mathcal{Y_N}}}, \hspace{0.075cm} \sigma \in \mathcal{L}_{N}^{\neg}\Rightarrow$ &  &  & \multirow{2}{*}{$\bullet$}  & \multirow{2}{*}{$\nu$}  & \multirow{2}{*}{$\nu$}  & \multirow{2}{*}{$\nu$}  &\\
$\sigma \wedge \varphi \in \mathcal{L}^{\prime}_{\mathcal{Y_N}}$ &  &  &  &  &  &  &\\
\cline{1-7}
$\varphi \in \mathcal{L}^{\prime}_{{\mathcal{Y_N}}}, \hspace{0.075cm} \sigma \in \mathcal{L}_{N}^{\equiv}\Rightarrow$ &  &  &  & \multirow{2}{*}{$\bullet$}  & \multirow{2}{*}{$\bullet$}  & \multirow{2}{*}{$\nu$}  &\\
$\sigma \wedge \varphi \in \mathcal{L}^{\prime}_{\mathcal{Y_N}}$ &  &  &  &  &  &  &\\
\cline{1-7}
$X\subseteq Act, \hspace{0.075cm} \varphi_a \in \mathcal{L}^{\prime}_{\mathcal{Y_N}} \hspace{0.075cm} \forall a \in X \Rightarrow$ &  &  &  &  & \multirow{2}{*}{$\bullet$} &\multirow{2}{*}{$\nu$} &\\
$\bigwedge_{a \in X} a\varphi_a \in \mathcal{L}^{\prime}_{\mathcal{Y_N}}$ &  &  &  &  &  &  &\\
\cline{1-7}
$\varphi_i \in \mathcal{L}^{\prime}_{{\mathcal{Y_N}}} \hspace{0.075cm} \forall i \in I \Rightarrow$ &  &  &  &  &  & \multirow{2}{*}{$\bullet$} &\\
$\bigwedge_{i \in I} \varphi_i \in \mathcal{L}^{\prime}_{{\mathcal{Y_N}}}$ &  &  &  &  &  &  &\\
\cline{1-7}
$\varphi \in \mathcal{L}_{N} \Rightarrow$ &  &  &  &  &  & \multirow{2}{*}{$\bullet$}   &\\
$\varphi \in \mathcal{L}^{\prime}_{\mathcal{Y_N}}$ &  &  &  &  &  &  &\\
\cline{1-7}
$\varphi \in \mathcal{L}_{N} \Rightarrow$ &  &  &  &  &  & \multirow{2}{*}{$\bullet$}  &\\
$\neg \varphi \in \mathcal{L}^{\prime}_{\mathcal{Y_N}}$ &  &  &  &  &  &  &\\
\cline{1-7}
\end{tabular}}
\vspace{1ex}
\caption{Our new logical characterizations for the semantics at each level of the ltbt-spectrum} \label{our_logic_table}
\vspace{-0.2cm}
\end{center}
\end{table}

In Tables \ref{our_constraint_table} and \ref{our_logic_table}, we present all our results in a three-dimensional way: Table \ref{our_logic_table} shows the rules defining the logics characterizing each of the semantics at each layer of the \emph{enlarged spectrum} (we provide an additional column with the particularization for $N=I$), while Table \ref{our_constraint_table} contains the logics that characterize the constraint governing each of these ``layers''.  As commented above, there are two semantics that appear in both tables, although disguised with different names: $T = \ \leq_U^{l}$ (in fact, it is also equal to the other three linear \emph{U-semantics}) and $S=US$.
\section{Relating the unified logics and the unified observational model}\label{logic_observational_framework}
In this Section we will relate the unified logical characterizations and the unified observational semantics developed in \cite{fgp09_observational}. As we indicated in Section \ref{preliminaries}, we have to restrict ourselves to finite image processes to obtain the result. As a byproduct, we get for this kind of processes that the finite parts of each of the corresponding languages, that are obtained by intersection with $\mathcal{L}^{f}_{HM}$, give us a pure finite logical characterization of the semantics.
We start by considering the following concept of normal formula.

\begin{definition}[\textbf{Normal formula $\mathcal{N(L)}$}]\label{def:normal}
\begin{enumerate}
\item Given a set of formulas $\mathcal{L}$, whose outermost operator is not the conjunction, we define the set of induced normal formulas, $\mathcal{N(L)}$, starting with $\top$ and adding those formulas that can be generated by applying the clause: If $\Gamma_1, \Gamma_2\subseteq \mathcal{L},\hspace{0.2cm} \{a_i \mid i\in I\} \subseteq Act$ and $\varphi_i \in \mathcal{N(L)}$, then $(\bigwedge_{\sigma\in \Gamma_1}\sigma \wedge \bigwedge_{\sigma \in \Gamma_2} \neg \sigma) \wedge \bigwedge_{i\in I}a_i \varphi_i \in \mathcal{N(L).}$
\item Now, for each $N \in \{U,C,I,T,S\}$ and each $\mathcal{Y_N} \in \{NS, \leq_N^{l}, \leq_N^{l\supseteq}, \leq_N^{lf}, \leq_N^{lf\supseteq}, \leq_N^{l\subseteq}, \leq_N^{lf\subseteq}, D_N\}$ in the spectrum, we define the set of normal formulas, $\mathcal{N_{Y_N}(L_N^{\prime \prime})} \subseteq \mathcal{L}^{\prime}_{\mathcal{Y_N}}$ simply as: $\mathcal{N_{Y_N}(L_N^{\prime \prime})}=\mathcal{N(L_N)}\bigcap \mathcal{L}^{\prime}_{\mathcal{Y_N}}$ where $\mathcal{L}_{N}^{\prime \prime}$ is the set of formulas in $\mathcal{L}^{\prime}_{N}$ whose outermost operator is not the conjunction.\end{enumerate}
\end{definition}

\begin{remark}
First note that the clause in Def. \ref{def:normal}.1 is a bit complicated: initially, we can apply it starting with $I=\emptyset$, and in this way we can obtain the first (non-trivial) normal formulas; then we can apply it recursively to obtain new, more complex, normal formulas; instead, the formulas in the two first subformulas come always from the original set $\mathcal{L}$.
Also note that we admit the use of infinite conjunction in those two first subformulas. As a consequence, these formulas could also have infinite depth (as infinite formulas in (the infinite generalizations of ) $\mathcal{L}_{HM}$). However, if we define the normal depth of formulas in $\mathcal{N(L_N)}$ as that obtained by counting the recursive nesting in the application of Def. \ref{def:normal}, then any normal formula has finite normal depth, and the set they form can be explored by structural induction.
\end{remark}

\begin{theorem}\label{equivalencia_logica_normal}
Each set of normal formulas $\mathcal{N_{Y_N}(L_N^{\prime \prime})}$ associated to each of the semantics in the spectrum  is equivalent to the full set of formulas $\mathcal{L}^{\prime}_{\mathcal{Y_N}}$.
\end{theorem}

\begin{definition}
We define the set of complete normal formulas $\mathcal{CN(L)}$ (resp. the set of complete normal formulas associated to each semantics in the spectrum, $\mathcal{CN_{Y_N}(L_N^{\prime \prime})}$) as the set of normal formulas (resp. the set of normal formulas associated to each semantics in the spectrum) that satisfy the condition $\Gamma_2= \overline{\Gamma_1}$, whenever the rule in Def. \ref{def:normal} is applied in the generation of each formula.
\end{definition}

Next we state that infinite conjunction in Def. \ref{def:normal} can be approximated by finite conjunction.

\begin{theorem}\label{aproximacion}
If we restrict ourselves to finite image processes, any complete normal formula $\varphi \in \mathcal{CN(L)}$ can be approximated by  a set of finite normal formulas $\{ \varphi^{k} \mid k \in \mathbb{N}\}$ that only use finite conjunction, that is, we have $p\models \varphi \Leftrightarrow p \models \varphi^{k}$ $\forall k \in \mathbb{N}$.
\end{theorem}

\begin{theorem}\label{isomorfia_observaciones}
We can define a natural correspondence between the set of complete normal formulas associated to a semantics $\mathcal{CN_{Y_N}(L_N^{\prime \prime})}$ and the corresponding domain of observations $BGO_N$ or $LGO_N$. That correspondence $\leftrightarrow$ satisfies that $\varphi \leftrightarrow \theta \Rightarrow (p \models \varphi \Leftrightarrow \theta \in XGO_N(p))$ with $X=B$ or $X=L\,$.
Moreover, this correspondence produces the following results for each of the semantics in the spectrum:
\begin{enumerate}
\item The set of complete normal formulas $\mathcal{CN_{NS}(L_N^{\prime \prime})}$ (resp. $\mathcal{CN_{D_N}(L_N^{\prime \prime})}$) and the domain of branching general observations $GBO_N$ (resp. $dBGO_N$) are isomorphic, that is, $\leftrightarrow$ is one to one.
\item The set of complete normal formulas $\mathcal{CN}_{\leq_N^{l}}\mathcal{(L_N^{\prime \prime})}$, $\mathcal{CN}_{\leq_N^{l\supseteq}}\mathcal{(L_N^{\prime \prime})}$ and the domain of linear general observations $LGO_N$ are isomorphic, that is, $\leftrightarrow$ is one to one.
\item The set of complete normal formulas $\mathcal{CN}_{\leq_N^{lf}}\mathcal{(L_N^{\prime \prime})}$ (resp. $\mathcal{CN}_{\leq_N^{lf\supseteq}}\mathcal{(L_N^{\prime \prime})}$) and the quotient domain $LGO_N/ _{\simeq_{N}^{lf}}$ (resp. $LGO_N/ _{\simeq_{N}^{lf\supseteq}}$) are isomorphic, that is, $\leftrightarrow^{-1}$ is injective and $\varphi \leftrightarrow \theta$ iff $\theta \simeq_N^{lf\supseteq} \theta_{\varphi}$, for some adequate $\theta_{\varphi}$.
\end{enumerate}
\end{theorem}

\begin{theorem}
The logical semantics   $\sqsubseteq^{\prime}_{\mathcal{Y_N}}$ induced by the logic $\mathcal{L}^{\prime}_{\mathcal{Y_N}}$, where $\mathcal{Y_N} \in \{NS, \leq_N^{l}, \leq_N^{l\supseteq}, \leq_N^{lf}, \leq_N^{lf\supseteq},$ $D_N\}$, is equivalent to the corresponding observational semantics, defined at Def. \ref{def:ramificadas} and Def. \ref{def:lineales}. 
\end{theorem}
\section{The real diamond structure} \label{cap:diamante}
Now we will explore in more detail the real structure of the extended spectrum, as it was already done at \cite{fgp09_equational}. One could think that each diamond in that spectrum corresponds to a lattice structure.  However, this is not the case: there is another semantics coarser than both \textit{N}-readiness and \textit{N}-failure traces and finer than \textit{N}-failures, and another finer than those two semantics and coarser than \textit{N}-ready traces.  

\begin{figure}[ht]
\begin{center}
\scalebox{0.9}{\scriptsize
\begin{picture}(147,43)
\put(0,19){RS} \put(30,19){PW} \put(60,19){RT} \put(106,0){R} \put(104,39){FT} \put(82,20){$R\wedge FT$} \put(112,19){$R \vee FT$} \put(147,19){F}
\put(13,22){\vector(1,0){15}} \put(43,22){\vector(1,0){15}} \put(72,25){\vector(2,1){31}} \put(72,19){\vector(2,-1){31}}\put(72,22){\vector(1,0){9}} \put(135,22){\vector(1,0){11}} \put(115,41){\vector(2,-1){31}}\put(114,3){\vector(2,1){32}}
\put(94,25){\vector(1,1){12}} \put(94,18){\vector(1,-1){11}}  \put(109,7){\vector(1,1){11}} \put(110,37){\vector(1,-1){12}}
\end{picture}}
\caption{The diamond below ready simulation}\label{diamante}
\vspace{-0.2cm}
\end{center}
\end{figure}
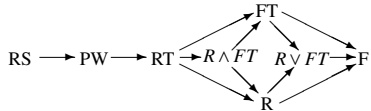

Focusing on the case $N=I$ the obtained complete structure is that shown in Figure \ref{diamante}, in which we include the new join semantics $R \wedge FT$ and the meet one $R \vee FT$.  As proved in \cite{fgp09_observational}, the meet semantics $R \vee FT$ was already studied by Roscoe under the name of revivals semantics in \cite{Ros09}. 

Since Readiness semantics observes the ready set at the end of the trace, while Failure Traces observes failures during the computation, it is natural to expect that the join semantics $R \wedge FT$  will observe both failures during the computation and ready sets at the end. This is indeed the case. The corresponding observational characterization in the general case is obtained by means of a new order $\leq_N^{l\supseteq \wedge f}$ on $LGO_N$. 

\begin{definition}\label{orden_wedge}
Let $\zeta,\zeta^\prime\subseteq LGO_N$, we define
\begin{center}
$\zeta\leq_N^{l\supseteq \wedge f}\zeta^\prime$ $\Leftrightarrow$ $\forall$ $X_0a_1X_1\ldots X_n \in \zeta$ $\exists$ $Y_0a_1Y_1\ldots Y_n \in \zeta^\prime$ $(\forall i \in 0..n-1$ $X_i\supseteq Y_i)$ $\wedge$ $X_n=Y_n\ .$
\end{center}
\end{definition}

It is easy to see that $\leq_N^{l\supseteq \wedge f}$ is indeed the conjunction of $\leq_N^{l\supseteq}$ and $\leq_N^{lf}$, that is, $\zeta\leq_N^{l\supseteq \wedge f}\zeta^\prime$ $\Leftrightarrow$ $\zeta\leq_N^{l\supseteq}\zeta^\prime \wedge$ $\zeta\leq_N^{lf}\zeta^\prime$.
The observational characterization of the meet semantics $R \vee FT$ is a bit more complicated.

\begin{definition}
Let $\zeta,\zeta^\prime\subseteq LGO_N$, we define
\begin{center}
$\zeta\leq_N^{l\supseteq\vee f}\zeta^\prime$ $\Leftrightarrow$ $\forall$ $X_0a_1X_1\ldots X_n \in \zeta$ $\exists$ $\{Y_0a_1Y_1\ldots Y_n^{j} |j \in J\} \subseteq \zeta^\prime$ such that $X_n=\bigcup_{j \in J}Y_n^{j}\ .$
\end{center}
\end{definition}

By means of some simple algebraic manipulations we can get the following equivalent expression:
\begin{center}
$\zeta\leq_N^{l\supseteq\vee f}\zeta^\prime$ $\Leftrightarrow$ $\forall$ $X_0a_1X_1\ldots X_n \in \zeta$ $\forall a \in X_n$ $\exists$ $Y_0a_1Y_1\ldots Y_n \in \zeta^\prime$ such that $(a \in Y_n \wedge Y_n \subseteq X_n)\ .$
\end{center}

Next we present the logical characterizations of these new semantics. Obviously, they are in the linear side of the spectrum and therefore they will have a similar structure to those for the linear semantics studied before. Once again, we start with the particular case $N=I$. $R \wedge FT$ is finer than both R and FT, and the logic characterizing it will be just the union of those characterizing $R$ and $FT$. In the case of $R \vee FT$ we need to connect the clauses that define those two logics in an adequate way. 

\begin{definition}\label{wedge_vee_logica}
\begin{enumerate}
\item We define the set of formulas $\mathcal{L}^{\prime}_{\leq_I^{l\supseteq\wedge f}}$, as that generated by the clauses: \lin{\conjyformneg{\leq_I^{l\supseteq\wedge f}}{I} ; \conj{\leq_I^{l\supseteq\wedge f}}{I}}{\leq_I^{l\supseteq\wedge f}}
\item We define the set of formulas $\mathcal{L}^{\prime}_{\leq_I^{l\supseteq\vee f}}$ as that generated by the clauses: \lin{\conjU{\leq_I^{l\supseteq\vee f}}{I}}{\leq_I^{l\supseteq\vee f}}
\end{enumerate}
\end{definition}

\begin{example}
$P_2$ and $P_3$ in Figure \ref{example1} satisfy $P_2 \sim_{F} P_3$, but $P_2 \nleq_{R \vee F} P_3 \ $. Taking p= abc+a(bd+c) and q= p+ a(bc+c) we have $p\sim_{R\wedge FT}q$ but $p\nsim_{RT}q\ $.
\end{example}

\begin{theorem}
The logical semantics  $\sqsubseteq^{\prime}_{\leq_I^{l\supseteq\wedge f}}$ (resp. $\sqsubseteq^{\prime}_{\leq_I^{l\supseteq\vee f}}$) induced by the logic $\mathcal{L}^{\prime}_{\leq_I^{l\supseteq\wedge f}}$ (resp. $\mathcal{L}^{\prime}_{\leq_I^{l\supseteq\vee f}}$) is equivalent to the observational semantics defined by $LGO_I$, with the order ${\leq_I^{l\supseteq\wedge f}}$ (resp. ${\leq_I^{l\supseteq\vee f}}$.)
\end{theorem}
\begin{proof}
In the case of $R\wedge FT$ we just need to check that $\mathcal{L}^{\prime}_{\leq_I^{l\supseteq\wedge f}}$ = $\mathcal{L}^{\prime}_{\leq^{l\supseteq}_I} \cup \mathcal{L}^{\prime}_{\leq^{lf}_{I}}$. The meet of two semantics is not always defined by the intersection of the corresponding logics. However, in this case we have that $\mathcal{L}^{\prime}_{\leq_I^{l\supseteq\vee f}}$ = $\mathcal{L}^{\prime}_{\leq^{l\supseteq}_I} \cap \mathcal{L}^{\prime}_{\leq^{lf}_{I}}$, and then to check that it defines $R \vee FT$ it is enough to see that $p \nsqsubseteq^{l\supseteq \vee f}_{I} q \ \Rightarrow \ (\exists \varphi \in \mathcal{L}^{\prime}_{\leq^{l\supseteq \vee f}_{I}} \ p\models \varphi \wedge q \nvDash \varphi)$, which is nearly immediate.
\end{proof}

By replacing \textit{I} above by the generic \textit{N} we get the definitions and results for the general case.
\section{Conclusions and future work}\label{conclusiones}
We have concluded in this paper the work on unification of  all the strong process semantics by considering here the logic approach, while \cite{fgp09_equational,fgp09_observational} considered the observational and the equational approaches. As in the previous cases, our main goal was to clarify the relationships between all the process semantics, that were classified in a slightly messy way in \cite{Gla01}. Our starting point has been the Hennessy-Milner Logic \cite{hm85}: we have looked for sublogics with a simple structure, that characterize each of the semantics in the \emph{enlarged spectrum}. The difference between branching-time semantics and linear-time semantics is the key point to isolate the ingredients that, combined in different ways, produce the different semantics. 

It is interesting to comment on the difference between the observational and the logical characterizations. Note that in the observational framework the observations had a complex structure, where local observations informed us about the (static) properties of the states of a process, while the arcs gave us the dynamic information. Instead, the formulas of the logic \emph{HML} do not possess of such structure, having only a \textit{low level} structure induced by the combination of  prefix and conjunction. This is why we needed to introduce normal forms in order to build the \textit{high level} structure of observations at the formulas.

Came as a surprise to us the discovery of two more linear semantics at each layer of the spectrum. Moreover, we found out that the classic logical characterization of Possible Worlds (PW) was wrong. A too ad-hoc selection of the rules defining each logic was probably the cause, that we discovered when trying to unfold the original characterization to look for the equivalent presentation inside our model.

Now that we have available all the unified characterizations of the semantics we have a much clearer picture of the spectrum, and we can use the parameterized definitions to prove generic properties of all or a part of the semantics in a generic way, without having to repeat similar proofs for each of them.

There are several directions in which we plan to extend our work. Weak semantics are an obvious target: if there are indeed many strong process semantics, once we introduce internal actions a terrible explosion occurs \cite{Gla93}, and the unification work is even more necessary  in order to clarify which are the most interesting semantics and what the differences between them are. Another interesting direction comes from the combinations of logic and algebra, as done by Luttgen and Vogler \cite{lv10,lv09}. Again, we are interested in studying whether their proposal is canonical or can be parameterized in some way in order to obtain other interesting combinations. Finally, a couple of papers \cite{bc10,Gut09} have appeared recently, where the logical characterizations of the non-interleaving semantics are developed.

\bibliographystyle{eptcs}
\bibliography{bibliografia}
\end{document}